\documentclass{llncs}

\usepackage[dvips]{graphicx}
\usepackage{latexsym}
\usepackage{pstricks,psfrag}
\usepackage{amsmath,amssymb}
\usepackage{algorithm}
\usepackage{algpseudocode}
\usepackage{url}

\renewcommand{\algref}[1]{Algorithm~{\rm\ref{alg:#1}}}
\newcommand{\figref}[1]{Fig.~{\rm\ref{fig:#1}}}
\newcommand{\tabref}[1]{Table~{\rm\ref{tab:#1}}}

\newcommand{\propref}[1]{Proposition~{\rm\ref{prop:#1}}}
\newcommand{\secref}[1]{Sect.~\ref{sec:#1}}
\newcommand{\thmref}[1]{Theorem~\ref{thm:#1}}
\renewcommand{\eqref}[1]{\rm{(\ref{eq:#1})}}

\newcommand{\QED}{\quad$\Box$} 

\newtheorem{defn}{Definition}
\newtheorem{prop}{Proposition}
\newtheorem{thm}{Theorem}
\newcommand{\LSSOL}{\textsc{LocalSolver}}
\newcommand{\MINISAT}{\textsc{MiniSat}}
\newcommand{\SUGAR}{\textsc{Sugar}}

\psfrag{[lx2]}{$\ell_{x,2}$}
\psfrag{[lx3]}{$\ell_{x,3}$}
\psfrag{[ly1]}{$\ell_{y,1}$}
\psfrag{[ly3]}{$\ell_{y,3}$}

\psfrag{[uxy]}{$u_{xy}$}
\psfrag{[uxz]}{$u_{xz}$}
\psfrag{[uyz]}{$u_{yz}$}
\psfrag{[lz1]}{$\ell_{z,1}$}
\psfrag{[lz2]}{$\ell_{z,2}$}

\long\def\invis#1{}

\title{An Efficient Local Search for\\ Partial Latin Square Extension Problem\thanks{This work is partially supported by JSPS KAKENHI Grant Number 25870661.}}
\author{Kazuya Haraguchi}
\institute{Faculty of Commerce, Otaru University of Commerce, Japan\\ \email{haraguchi@res.otaru-uc.ac.jp}
}

\begin{document}
\maketitle
\begin{abstract}
  A {\em partial Latin square\/} ({\em PLS\/}) is a partial assignment of $n$ symbols to
  an $n\times n$ grid such that,
  in each row and in each column, each symbol appears at most once.
  The {\em partial Latin square extension\/} problem
  is an NP-hard problem that asks for a largest extension of a given PLS.
  In this paper
  we propose an efficient {\em local search\/} for this problem. 
  We focus on the local search
  such that the neighborhood is defined by {\em $(p,q)$-swap\/}, i.e., 
  removing exactly $p$ symbols and then assigning symbols to at most $q$ empty cells.
  For $p\in\{1,2,3\}$, our neighborhood search algorithm
  finds an improved solution or concludes that no such solution exists in $O(n^{p+1})$ time. 
  We also propose a novel swap operation, Trellis-swap,
  which is a generalization of $(1,q)$-swap and $(2,q)$-swap.
  Our Trellis-neighborhood search algorithm takes $O(n^{3.5})$ time
  to do the same thing. 
  Using these neighborhood search algorithms, 
  we design a prototype iterated local search algorithm
  and show its effectiveness in comparison with 
  state-of-the-art optimization solvers such as IBM ILOG CPLEX and \LSSOL. 
\end{abstract}

\section{Introduction}
\label{sec:intro}

We address the
{\em partial Latin square extension\/} ({\em PLSE\/}) problem. 
Let $n\ge2$ be a natural number. 
Suppose that we are given an $n\times n$ grid of {\em cells\/}. 
A {\em partial Latin square\/} ({\em PLS\/}) is a partial assignment of {\em $n$ symbols\/}
to the grid so that the {\em Latin square condition\/} is satisfied.
The Latin square condition requires that,
in each row and in each column,
every symbol should appear at most once. 
Given a PLS, the PLSE problem asks to fill as many empty cells with symbols as possible 
so that the Latin square condition remains to be satisfied.

In this paper, we propose an efficient {\em local search\/} for the PLSE problem. 
Let us describe our research background and motivation.
The PLSE problem is practically important
since it has various applications such as combinatorial design,
scheduling, optical routers, and 
combinatorial puzzles~\cite{BA.1993,CD.2006,GS.2002}. 
The problem is NP-hard~\cite{C.1984}, and was first studied by Kumar et al.~\cite{KRS.1999}.
The problem has been studied 
in the context of constant-ratio approximation algorithms~\cite{GRS.2004,HJKS.2007,HO.2014,KRS.1999}.
Currently the best approximation factor
is achieved by a local search algorithm~\cite{C.2013,FH.2014,HJKS.2007}. 
In that local search the neighborhood is defined by {\em $(p,q)$-swap\/},
where $p$ and $q$ are non-negative integers such that $p<q$. 
It is the operation of removing exactly $p$ symbols 
from the current solution and then assigning symbols to at most $q$ empty cells. 
To the best of the author's knowledge, 
there is no literature that investigates efficient implementation of local search. 

Our local search is based on Andrade et al.'s local search~\cite{ARW.2012}
for the {\em maximum independent set\/} ({\em MIS\/}) problem. 
The MIS problem is a well-known NP-hard problem as well~\cite{GJ.1979}. 
We utilize Andrade et al.'s methodology
since, as we will see later,
the PLSE problem is a special case of the MIS problem.

We improve the efficiency of the local search
by utilizing the problem structure peculiar to the PLSE problem. 
Specifically, for $p\in\{1,2,3\}$ and $q=n^2$,
our neighborhood search algorithm takes only $O(n^{p+1})$ time
to find an improved solution 
or to conclude that no improved solution exists in the neighborhood,
whereas the direct usage of Andrade et al.'s algorithm ($p=1$ and 2)
and Itoyanagi et al.'s algorithm $(p=3)$~\cite{IHY.2011}
requires $O(n^{p+3})$ time to do the same things even for $q=p+1$.
Note that $q=n^2$ is the upper limit of
the number of nodes that can be inserted in a swap operation.
Our swap operations insert as many nodes to the solution as possible.   

We then propose a new type of swap operation that we call {\em Trellis-swap\/}. 
It is a generalization of $(1,n^2)$-swap and $(2,n^2)$-swap,
and contains certain cases of $(3,n^2)$-swap.
Our Trellis-neighborhood search algorithm
takes $O(n^{3.5})$ time to find an improved solution,
or to conclude that no improved solution exists in the neighborhood.

We regard our local search as efficient
since the time complexities above should be the best possible bounds.  
For example, when $p=1$, we may not be able to improve the bound $O(n^2)$
further
since in fact the bound is linear with respect to the solution size. 

We show how our local search is effective through computational studies. 
The highlight is that our prototype {\em iterated local search\/} ({\em ILS\/}) algorithm
is likely to deliver a better solution
than such state-of-the-art optimization softwares 
as IP and CP solvers from IBM ILOG CPLEX~\cite{CPLEX} and
a general heuristic solver from \LSSOL~\cite{LSSOL}.
Furthermore, among several ILS variants,
the best is one based on Trellis-swap. 

The decision problem version of the PLSE problem is known as 
the {\em quasigroup completion\/} ({\em QC\/}) problem 
in AI, CP and SAT communities~\cite{AVDFMS.2004,GS.1997,GS.2002,SATCOMP}. 
The QC problem has been one of the most frequently used benchmark problems
in these areas
and variant problems are studied intensively,
e.g., Sudoku~\cite{CACM.2008,CCM.2009,LMS.2006,L.2007,SIMONIS,SCGMP.2013},
mutually orthogonal Latin squares~\cite{AMM.2006,FJ.2013,VSKB.2011},
and spatially balanced Latin squares~\cite{GSEE.2004,BPG.2012,SGF.2005}. 
Our local search may be helpful for those who develop exact solvers for the QC problem
since the local search itself or metaheuristic algorithms employing it
would deliver a good initial solution or a tight lower estimate
of the optimal solution size quickly.

The paper is organized as follows.
In \secref{prel}, 
preparing terminologies and notations,
we see that the PLSE problem is a special case of the MIS problem. 
We explain the algorithms and the data structure of our local search in \secref{ls}
and then present experimental results in \secref{exp}.
Finally we conclude the paper in \secref{conc}.

\section{Preliminaries}
\label{sec:prel}

Let us begin with formulating the PLSE problem. 
Suppose an $n\times n$ grid of cells.
We denote $[n]=\{1,2,\dots,n\}$. 
For any $i,j\in[n]$, we denote the cell in the row $i$
and in the column $j$ by $(i,j)$. 
We consider a partial assignment of $n$ symbols to the grid. 
The $n$ symbols to be assigned are $n$ integers in $[n]$. 
We represent a partial assignment by a set of triples, say $T\subseteq[n]^3$, such that
the membership $(v_1,v_2,v_3)\in T$ indicates that the symbol $v_3$ is assigned to $(v_1,v_2)$. 
To avoid a duplicate assignment, we assume that,
for any two triples $v=(v_1,v_2,v_3)$ and $w=(w_1,w_2,w_3)$ in $T$ $(v\ne w)$,
$(v_1,v_2)\ne(w_1,w_2)$ holds. 
Thus $|T|\le n^2$ holds. 

For any two triples $v,w\in[n]^3$, we denote the Hamming distance between $v$ and $w$ by $\delta(v,w)$, 
%
%
i.e., $\delta(v,w)=|\{k\in[3]\mid v_k\ne w_k\}|$.  
We call a partial assignment $T\subseteq[n]^3$
a {\em PLS set\/} if, for any two triples $v,w\in T$ $(v\ne w)$,
$\delta(v,w)$ is at least two. 
One easily sees that $T$ is a PLS set
iff it satisfies the Latin square condition.
%
%
%
We say that two disjoint PLS sets $S$ and $S'$ are {\em compatible\/}
if, for any $v\in S$ and $v'\in S'$,
the distance $\delta(v,v')$ is at least two.
Obviously the union of such $S$ and $S'$ is a PLS set. 
The PLSE problem is then formulated as follows; given a PLS set $L\subseteq[n]^3$,
%
%
we are asked to construct a PLS set $S$ of the maximum cardinality
such that $S$ and $L$ are compatible.

Next, we formulate the MIS problem. 
An {\em undirected graph\/} (or simply a {\em graph\/}) $G=(V,E)$
consists of a set $V$ of {\em nodes\/}
and a set $E$ of unordered pairs of nodes, 
where each element in $E$ is called an {\em edge\/}. 
When two nodes are joined by an edge,
we say that they are {\em adjacent\/}, or equivalently, 
that one node is a {\em neighbor\/} of the other. 
An {\em independent set\/} is a subset $V'\subseteq V$ of nodes such that 
no two nodes in $V'$ are adjacent. 
Given $G$, the MIS problem asks for a largest independent set. 
For any node $v\in V$, we denote the set of its neighbors by $N(v)$.
The number $|N(v)|$ of $v$'-s neighbors is called
the {\em degree of $v$\/}.

Now we are ready to transform any PLSE instance into an MIS instance. 
Suppose that we are given a PLSE instance in terms of a PLS set $L\subseteq[n]^3$.
For any triple $v\in L$, we denote by $N^\ast(v)$
the set of all triples $w$'-s in the entire $[n]^3$
such that $\delta(v,w)=1$, i.e., $N^\ast(v)=\{w\in [n]^3\mid \delta(v,w)=1\}$. 
Clearly we have $|N^\ast(v)|=3(n-1)$. 
The union $\bigcup_{v\in L}N^\ast(v)$ over $L$ is denoted by $N^\ast(L)$. 
\begin{prop}
\label{prop:trans}
A set $S\subseteq[n]^3$ of triples 
is a feasible solution to the PLSE instance $L$
iff $S$, as a node set, is a feasible solution
to the MIS instance $G_L=(V_L,E_L)$ such that
$V_L=[n]^3\setminus(L\cup N^\ast(L))$ and  
$E_L=\{(v,w)\in V_L\times V_L\mid \delta(v,w)=1\}$.
\end{prop}
We omit the proof due to space limitation. 
By \propref{trans}, we hereafter consider solving the PLSE instance
by means of solving the transformed MIS instance $G_L=(V_L,E_L)$. 
Omitting the suffix $L$,
we write $G=(V,E)$ to represent $G_L=(V_L,E_L)$ for simplicity. 

Let us observe the structure of $G$. 
We regard each node $v=(v_1,v_2,v_3)\in V$ as a grid point in the 3D integral space. 
Any grid point is an intersection of three grid lines
that are orthogonal to each other. 
In other words, each node is on exactly three grid lines. 
A grid line is {\em in the direction $d$\/}
if it is parallel to the axis $d$ and
perpendicular to the 2D plane that is generated by
the other two axes. 
We denote by $\ell_{v,d}$ the grid line
in the direction $d$ that passes $v$.
Two nodes are joined by an edge iff
there is a grid line that passes both of them. 
The nodes on the same grid line form a {\em clique\/}. 
This means that any independent set should contain
at most one node among those on a grid line. 
Since $|N(v)|\le|N^\ast(v)|=3(n-1)$ and $|V|=O(n^3)$, we have $|E|=O(n^4)$.

We introduce notations and terminologies on local search for the MIS problem. 
We call any independent set in $G$ simply a {\em solution\/}. 
Given a solution $S\subseteq V$,
we call any node $x\in S$ a {\em solution node\/}
and any node $v\notin S$ a {\em non-solution node\/}. 
For a non-solution node $v$,
we call any solution node in $N(v)$ a {\em solution neighbor of $v$\/}. 
We denote the set of solution neighbors by $N_S(v)$, i.e., $N_S(v)=N(v)\cap S$. 
Since $v$ has at most one solution neighbor on one grid line
and three grid lines pass $v$, we have $|N_S(v)|\le3$. 
We call the number $|N_S(v)|$ the {\em tightness of $v$\/}
and denote it by $\tau_S(v)$. 
When $\tau_S(v)=t$, we call $v$ {\em $t$-tight\/}. 
In particular, a 0-tight node is called {\em free\/}. 
When $x$ is a solution neighbor of a $t$-tight node $v$,
we may say that $v$ is a {\em $t$-tight neighbor of $x$\/}.

For two integers $p,q$ such that $0\le p<q$, 
the {\em $(p,q)$-swap\/} refers to
an operation of removing exactly $p$ solution nodes from $S$
and inserting at most $q$ free nodes into $S$
so that $S$ continues to be a solution. 
The {\em $(p,q)$-neighborhood of $S$\/}
is a set of all solutions that are obtained 
by performing a $(p,q)$-swap on $S$. 
%
%
We assume $q\le n^2$
since, for any $q>n^2$,
the $(p,q)$-neighborhood and the $(p,n^2)$-neighborhood are equivalent. 
A solution $S$ is called {\em $(p,q)$-maximal\/} if the $(p,q)$-neighborhood 
contains no improved solution $S'$ such that $|S'|>|S|$. 
We call a solution {\em $p$-maximal\/} if it is $(p,n^2)$-maximal. 
In particular, we call a 0-maximal solution simply a {\em maximal\/} solution. 
Being $p$-maximal implies that $S$ is also $p'$-maximal for any $p'<p$.
Equivalently, if $S$ is not $p'$-maximal, then it is not $p$-maximal either for any $p>p'$. 
%

\section{Local Search}
\label{sec:ls}
In this section, we present the main component algorithm of the local search.
The main component is a {\em neighborhood search algorithm\/}. 
Given a solution $S$ and a neighborhood type being specified, 
it computes an improved solution in the neighborhood
or concludes that no such solution exists. 
Once a neighborhood search algorithm is established, 
it is immediate to design a local search algorithm
that computes a maximal solution in the sense of the specified neighborhood type;
starting with an appropriate initial solution,
we repeat moving to an improved solution
as long as the neighborhood search algorithm delivers one.

Specifically we present $(p,n^2)$-neighborhood search algorithms $(p\in\{1,2,3\})$
and a Trellis-neighborhood search algorithm. 
The basic data structure is borrowed from \cite{ARW.2012},
but we improve the efficiency by using the problem structure
peculiar to the PLSE problem.
The $(p,n^2)$-neighborhood search algorithms run in $O(n^{p+1})$ time,
whereas Trellis-neighborhood search algorithm runs in $O(n^{3.5})$ time. 
We describe the data structure 
that is commonly used in all these algorithms 
in \secref{ls-data}
and present the neighborhood search algorithms in \secref{ls-swap}.
Finally in \secref{ls-time}, from the viewpoint of approximation algorithms, 
we mention approximation factors 
of $p$-maximal and Trellis-maximal solutions
and analyze the time complexities that are needed to compute them.

We claim that our work should be far from trivial. 
Without using the MIS formulation,
one may conceive a $(1,2)$-neighborhood
search algorithm that runs in $O(n^3)$ time,
but its improvement is not easy. 
Concerning previous local search algorithms for the MIS problem,
their direct usage would require more computation time.
Andrade et al.'s~\cite{ARW.2012} $(1,2)$-neighborhood search
algorithm (resp., $(2,3)$-neighborhood search algorithm)
requires $O(|E|)=O(n^4)$ time (resp., $O(\Delta|E|)=O(n^5)$ time,
where $\Delta$ denotes the maximum degree in the graph). 
Itoyanagi et al.~\cite{IHY.2011} extended Andrade et al.'s work 
to the {\em maximum weighted independent set problem\/}.
Their $(3,4)$-neighborhood search algorithm runs in $O(\Delta^2|E|)=O(n^6)$ time.

\subsection{Data Structure}
\label{sec:ls-data}

We mostly utilize the data structure of Andrade et al.'s~\cite{ARW.2012}.
We represent a solution by a permutation of nodes. 
In the permutation, every solution node is ordered ahead of all the non-solution nodes,
and among the non-solution nodes,
every free node is ordered ahead of all the non-free nodes. 
In each of the three sections (i.e., solution nodes, free nodes and non-free nodes),
the nodes can be ordered arbitrarily. 
We also maintain the solution size and the number of free nodes. 
For each non-solution node $v\notin S$,
we store its tightness $\tau_S(v)$
and pointers to its solution neighbors. 
Since $\tau_S(v)\le3$,
we store at most three pointers for $v$. 

Let us mention the novel settings
that we introduce additionally to enhance the efficiency. 
Regarding each node as a grid point in the 3D integral space,
we store the node set $V$ by means of a 3D $n\times n\times n$ array,
denoted by $C$.
For each triple $(v_1,v_2,v_3)\in[n]^3$,
if $(v_1,v_2,v_3)\in V$, then we let $C[v_1][v_2][v_3]$ have
the pointer to the node $(v_1,v_2,v_3)$,
and otherwise, we let it have a null pointer. 
For each solution node $x\in S$,
we store the number of its 1-tight neighbors 
along each of the three grid lines passing $x$. 
We denote this number by $\mu_d(x)$ $(d=1,2,3)$,
where $d$ denotes the direction of the grid line. 
We emphasize that maintaining $\mu_d(x)$ should play a key role
in improving the efficiency of the local search. 
Clearly the size of the data structure is $O(n^3)$. 
We can construct it in $O(n^3)$ time, as preprocessing of local search.

Using the data structure, we can execute some significant operations efficiently.
For example, we can identify whether $S$ is maximal or not in $O(1)$ time;
it suffices to see whether the number of free nodes is zero or not. 
The neighbor set $N(v)$ of a node $v=(v_1,v_2,v_3)$ 
can be listed in $O(n)$ time by searching $C[v'_1][v_2][v_3]$'-s,
$C[v_1][v'_2][v_3]$'-s and $C[v_1][v_2][v'_3]$'-s
for every $v'_1,v'_2,v'_3\in[n]$ such that 
$v'_1\ne v_1$, $v'_2\ne v_2$ and $v'_3\ne v_3$. 
Furthermore, we can remove a solution node from $S$ or 
insert a free node into $S$ in $O(n)$ time. 
We describe how to implement the removal operation. 
Let us consider removing $x\in S$ from $S$. 
For the permutation representation,
we exchange the orders between $x$ and the last node in the solution node section. 
We decrease the number of solution nodes by one
and increase the number of free nodes by one;
as a result, $x$ falls into the free node section. 
Its tightness is set to zero since it has no solution neighbor. 
For each neighbor $v\in N(x)$, 
we release its pointer to $x$ since $x$ is no longer a solution node,
and decrease the tightness $\tau_S(v)$ by one. 
\begin{itemize}
\item If $\tau_S(v)$ is decreased to zero,
  then $v$ is now free. 
  To put $v$ in the free node section, 
  we exchange the permutation orders 
  between $v$ and the head node in the non-free node section,
  and increase the number of free nodes by one. 
\item If $\tau_S(v)$ is decreased to one, then $v$ is now 1-tight
  and has a unique solution neighbor, say $y$.
  We increase the number $\mu_d(y)$ by one,
  where $d$ is the direction of the grid line that passes both $v$ and $y$. 
\end{itemize}
The total time complexity is $O(n)$. 
The insertion operation can be implemented in an analogous way.

\subsection{Neighborhood Search Algorithms}
\label{sec:ls-swap}

Let us describe the key idea on how to realize efficient neighborhood search algorithms. 
Suppose that a solution $S$ is given. 
Removing a subset $R\subseteq S$ from $S$ makes $R$ and certain neighbors free,
that is, non-solution nodes such that all solution neighbors are contained in $R$. 
If there is an independent set 
among these free nodes whose size is larger than $|R|$,
then there is an improved solution. 
Of course a larger independent set is preferred.
A largest one can be computed efficiently in some cases
although the MIS problem is computationally hard in general.

The first case is when $p=|R|$ is a small constant. 
Below we explain $(p,n^2)$-neighborhood search algorithms for $p=1$, 2 and 3.
Running in $O(n^{p+1})$ time, the algorithms have the similar structures.
Each algorithm searches all $R$'-s that can lead to an improved solution
by means of searching ``trigger'' nodes, at it were, 
sweeping the permutation representation.
%
The time complexity is linearly bounded by the number of trigger nodes; 
For each trigger node, the algorithm collects
certain solution nodes around it, which are used as $R$.
This takes $O(1)$ time.
On the other hand,
the algorithm does not search for the non-solution nodes to be inserted
unless it finds the MIS size among the free nodes from $S\setminus R$ larger than $|R|$.  
Surprisingly we can decide the MIS size in $O(1)$ time.
Only when the size is larger than $|R|$,
the algorithm searches for the MIS to be inserted,
and thereby it obtains an improved solution and terminates. 
The search for the MIS
requires $O(n)$ time, but it does not affect the total time complexity. 

Another case such that the MIS problem is solved efficiently
is when all solution nodes in $R$ are contained in such a 2D facet
that is induced by fixing the value of one dimension in the 3D space. 
In this case, the MIS problem is solved by means of bipartite maximum matching. 
This motivates us to invent a novel swap operation, Trellis-swap.  
In this swap, the size $|R|$ 
is not a constant but can change from 1 to $n$. 
Trellis-neighborhood search algorithm runs in $O(n^{3.5})$ time. 

\subsubsection{$(1,n^2)$-Swap.}
Let $S$ be a maximal solution.
%
%
%
Its $(1,n^2)$-neighborhood contains an improved solution iff 
there are a solution node $x\in S$ and non-solution nodes
$u,v\notin S$ such that
$(S\setminus\{x\})\cup\{u,v\}$ is a solution. 
It is clear that $u$ and $v$ should be neighbors of $x$.
They are $1$-tight, and their unique solution neighbor is $x$. 
The $u$ and $v$ should not be adjacent, 
which implies that $u$ and $v$ are not on the same grid line. 
Then, the number of nodes that can be inserted into $S\setminus\{x\}$
is given by 
$\nu(x)=|\{d\in\{1,2,3\}\mid \mu_d(x)>0\}|$,
i.e., the number of grid lines passing $x$
on which a 1-tight node exists.

\begin{thm}
  \label{thm:1}
  Given a solution $S$,
  we can find an improved solution in its $(1,n^2)$-neighborhood
  or conclude that it is 1-maximal in $O(n^2)$ time. 
\end{thm}
\begin{proof}
We assume that $S$ is maximal;
we can check in $O(1)$ time whether $S$ is maximal or not. 
If it is not maximal,
we have an improved solution by inserting any free node into $S$.
Finding a free node and inserting it into $S$
take $O(n)$ time, and then we have done. 

An improved solution exists iff there is $x\in S$ such that $\nu(x)\ge2$. 
All solution nodes can be searched by sweeping the first section of the permutation representation.
There are at most $n^2$ solution nodes. 
For each solution node $x$, the number $\nu(x)$ can be computed in $O(1)$ time.  
If $x$ with $\nu(x)\ge2$ is found,
we can determine the 1-tight nodes to be inserted in $O(n)$ time
by searching each grid line with $\mu_d(x)>0$.
Removing $x$ from $S$ and
inserting the nodes into $S\setminus\{x\}$ take $O(n)$ time. 
\QED
\end{proof}

\subsubsection{$(2,n^2)$-Swap.}
Let $S$ be a 1-maximal solution. 
Its $(2,n^2)$-neighborhood contains an improved solution
iff there exist $x,y\in S$ and $u,v,w\notin S$ 
such that $(S\setminus\{x,y\})\cup\{u,v,w\}$ is a solution. 
These nodes should satisfy Lemmas~1 to 4 in \cite{ARW.2012},
which are conditions established for the general MIS problem. 
According to the conditions, at least one node in $\{u,v,w\}$ is 2-tight,
and $x$ and $y$ are the solution neighbors of the 2-tight node. 
Let $u$ be this 2-tight node without loss of generality. 

See \figref{twoswap}. 
Suppose inserting $u$ into $S\setminus\{x,y\}$.
The nodes that can be inserted additionally should be on the four solid grid lines,
say $\ell_{x,2}$, $\ell_{x,3}$, $\ell_{y,1}$ and $\ell_{y,3}$.  
Let $F$ be the set of nodes on these grid lines
that are free from $(S\cup\{u\})\setminus\{x,y\}$. 
Let $v$ be any node in $F$. 
For $S$, $v$ should not have any solution neighbors other than $x$ or $y$
since otherwise it would not be free from $(S\cup\{u\})\setminus\{x,y\}$. 
Then we have $N_S(v)\subseteq\{x,y\}$ and thus $v$ is either 1-tight or 2-tight.  
As can be seen, $u_{xy}$,
the intersection point of $\ell_{x,2}$ and $\ell_{y,1}$,
is the only possible 2-tight node. 
All the other nodes in $F$ are 1-tight.
Note that, even though the node $u_{xy}$ exists,
it does not necessarily belong to $F$; it can be 3-tight. 

\begin{figure}[t]
  \centering
  \begin{tabular}{c}
    \includegraphics[width=5.0cm,keepaspectratio]{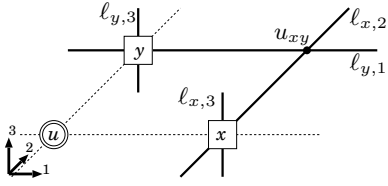}
  \end{tabular}
  \caption{An illustration of the case such that $(2,n^2)$-swap can occur}
  \label{fig:twoswap}
\end{figure}

Suppose that the node $u_{xy}$ exists and it is 2-tight.
In the subgraph induced by $F$, 
one can readily see that 
{\em every\/} MIS contains $u_{xy}$
only when no 1-tight node exists on $\ell_{x,2}$ or $\ell_{y,1}$
(i.e., $\mu_2(x)=\mu_1(y)=0$).
Thus, when we count the MIS size in the subgraph,
we need to take $u_{xy}$ into account only in this case. 

\begin{thm}
  \label{thm:2}
  Given a solution $S$,
  we can find an improved solution in its $(2,n^2)$-neighborhood
  or conclude that it is 2-maximal in $O(n^3)$ time. 
\end{thm}
\begin{proof}
Similarly to the proof of \thmref{1},
we assume that $S$ is 1-maximal. 

Let $u$ be any 2-tight node and $x$ and $y$ be its solution neighbors. 
The number of 2-tight nodes is at most $|V\setminus S|\le|V|\le n^3$. 
We can recognize its solution neighbors $x$ and $y$ in $O(1)$ time by tracing the pointers from $u$.
The size of an MIS among $F$ can be computed in $O(1)$ time as follows;
First we count the number of the four grid lines
$\ell_{x,2}$, $\ell_{x,3}$, $\ell_{y,1}$ and $\ell_{y,3}$
such that a 1-tight node exists,
where we follow the dimension indices in \figref{twoswap}
without loss of generality. 
This can be done by checking whether $\mu_d(x)>0$ (or $\mu_d(y)>0$) or not. 
Furthermore, if $\mu_2(x)=\mu_1(y)=0$,
we check whether a 2-tight node $u_{xy}$ exists or not. 
The check can be done in $O(1)$ time by referring to the 3D array $C$. 
If it exists, we increase the MIS size by one. 
Finally, if the MIS size is no less than two,
there exists an improved solution. 
The non-solution nodes to be inserted other than $u$
are found in $O(n)$ time by searching the four grid lines.
\QED
\end{proof}

\subsubsection{$(3,n^2)$-Swap.}
Given a solution $S$, 
consider searching for an improved solution in its $(3,n^2)$-neighborhood.
To reduce the search space the following result on the general MIS problem 
is useful. 
\begin{thm}{\bf (Itoyanagi et al.~\cite{IHY.2011})}
\label{thm:ito}
Suppose that $S$ is a 2-maximal solution
and that there are a subset $R\subseteq S$ of solution nodes
and a subset $F\subseteq V\setminus S$ of non-solution nodes 
such that $|R|=3$, $|F|=4$ and $(S\setminus R)\cup F$ is a solution. 
We denote $F=\{u,v,w,t\}$, and without loss of generality,
we assume $\tau_S(u)\ge\tau_S(v)\ge\tau_S(w)\ge\tau_S(t)$. 
Then we are in either of the following two cases:
\begin{description}
\item[(I)] $u$ is 3-tight and $R=N_S(u)$.
\item[(II)] $u$ and $v$ are 2-tight
  such that they have exactly one solution node as a common solution neighbor,
  i.e., $|N_S(u)\cap N_S(v)|=1$, and $R=N_S(u)\cup N_S(v)$. 
\end{description}
\end{thm}
\begin{thm}
  \label{thm:3}
  Given a solution $S$,
  we can find an improved solution in its $(3,n^2)$-neighborhood
  or conclude that it is 3-maximal in $O(n^4)$ time. 
\end{thm}
\begin{proof}
We assume that $S$ is 2-maximal, similarly to previous theorems. 
Below we prove that the search for (I) in \thmref{ito}
takes $O(n^3)$ time, and that for (II) takes $O(n^4)$ time.

For (I), we search all 3-tight nodes. There are $O(n^3)$ 3-tight nodes.
For each 3-tight node $u$, the solution nodes to be removed
are the three solution neighbors of $u$, 
which can be decided in $O(1)$ time.
Let $x$, $y$ and $z$ be the three solution neighbors,
and $F$ be the set of free nodes from $(S\setminus\{x,y,z\})\cup \{u\}$. 
The nodes in $F$ should be on solid grid lines in \figref{threeswap}~(I). 
There are at most three 2-tight nodes among $F$,
that is, $u_{xy}$, $u_{xz}$ and $u_{yz}$. 
Similarly to the proof of \thmref{2}, every MIS among the free nodes
contains $u_{xy}$ only when no 1-tight node exists on $\ell_{x,2}$ or $\ell_{y,1}$,
i.e., $\mu_2(x)=\mu_1(y)=0$. 
The condition on which $u_{xz}$ (or $u_{yz}$)
belongs to every MIS is analogous. 
Taking these into account, 
we can decide the MIS size among the free nodes in $O(1)$ time.
If it is no less than three, then an improved solution exists;
an MIS to be inserted can be decided in $O(n)$ time. 

\begin{figure}[t]
  \centering
  \begin{tabular}{ccc}
    \includegraphics[width=5.0cm,keepaspectratio]{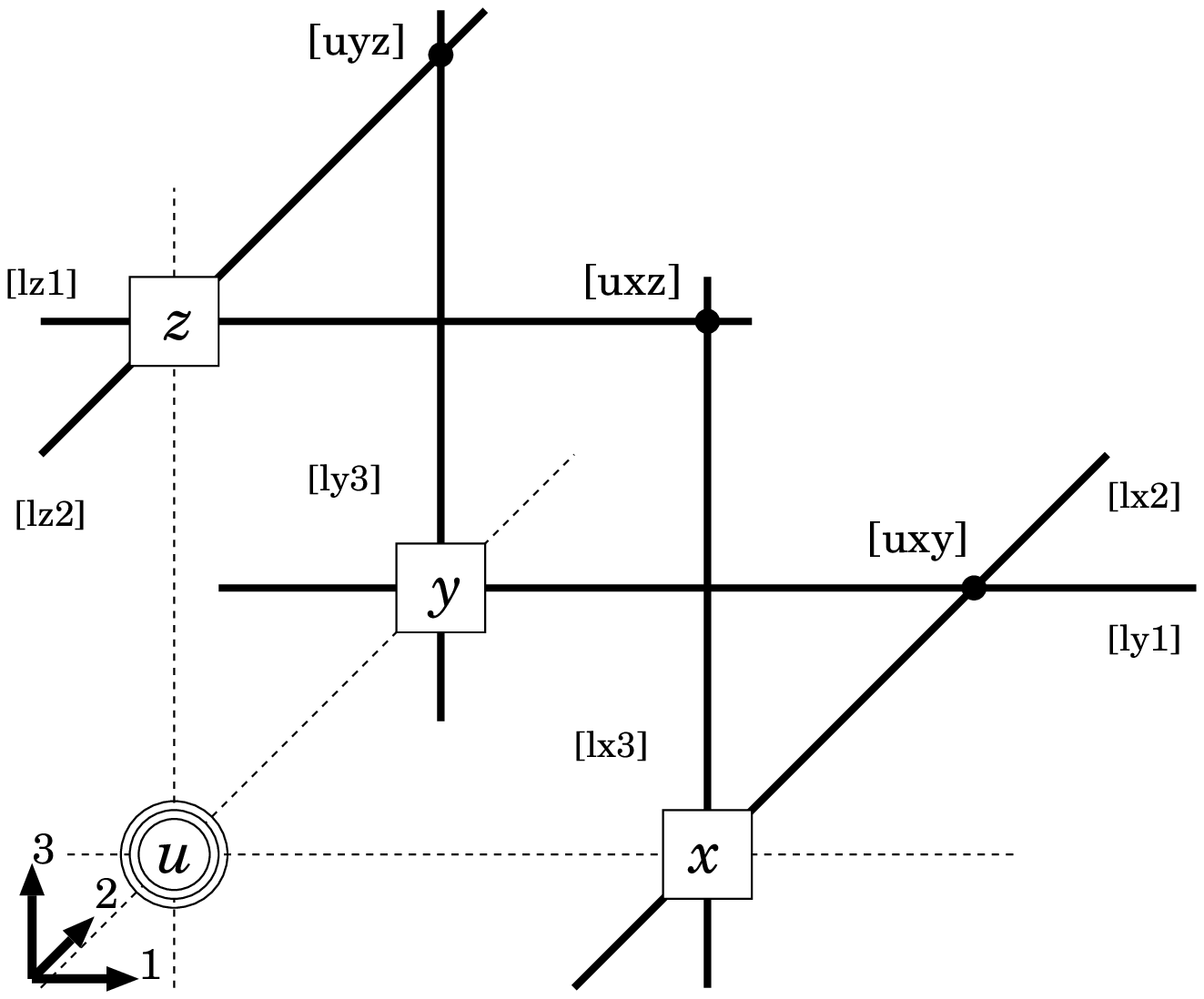} &\ \ \ &
    \includegraphics[width=5.0cm,keepaspectratio]{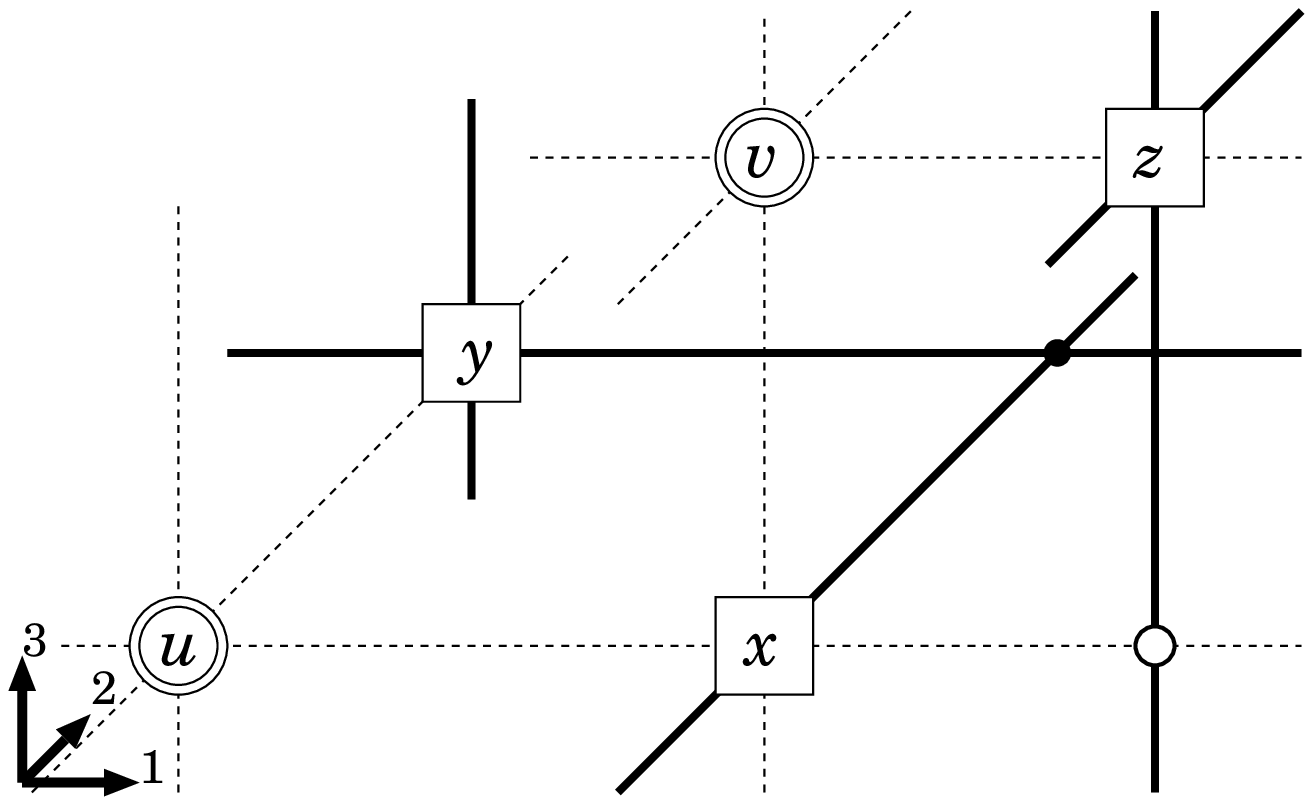}\\
    {\bf (I)} && {\bf (II)}
  \end{tabular}
  \caption{An illustration of the cases (I) and (II) such that $(3,n^2)$-swap can occur}
  \label{fig:threeswap}
\end{figure}

Concerning (II),
we search every solution node $x\in S$ and
every pair of its 2-tight neighbors that are not on the same grid line. 
Thus there are $O(n^4)$ pairs in all 
since there are at most $n^2$ solution nodes, and for each $x\in S$,
there are $O(n^2)$ pairs of 2-tight neighbors. 
Given $x$ and a pair $\{u,v\}$ of its 2-tight neighbors,
the removed solution nodes are three solution nodes in $N_S(u)\cup N_S(v)$, 
which includes $x$. They can be found in $O(1)$ time. 
Let $N_S(u)\cup N_S(v)=\{x,y,z\}$. 
We illustrate an example of this case in \figref{threeswap}~(ii). 
The free nodes from $(S\setminus\{x,y,z\})\cup\{u,v\}$
should be on the solid grid lines in the figure. 
By means of enumerative argument,
we can show that the MIS size among the free nodes
can be decided in $O(1)$ time. 
The argument is rather complicated, and we omit details here. 
If the MIS size is no less than two, then an improved solution exists;
an MIS to be inserted can be decided in $O(n)$ time. 
\QED
\end{proof}

\subsubsection{Trellis-Swap.}
Now let us introduce Trellis-swap. 
Note that, in $(2,n^2)$-swap,
all removed solution nodes belong to such a 2D facet
that is induced by fixing the value of one dimension in the 3D space.
We can regard that $(1,n^2)$-swap is also in the case.
We invent a more general swap operation for such a case. 

\begin{defn}
Suppose that a solution $S$ is given. 
Let $R\subseteq S$ be a maximal subset such that
all nodes in $R$ are contained in a 2D facet
that is generated by setting $v_d=k$ for some $d\in[3]$ and $k\in[n]$. 
We denote by $F_1$ (resp., $F_2$)
the set of all 1-tight (resp., 2-tight) nodes 
such that the unique solution neighbor is contained in $R$
(resp., both of the solution neighbors are contained in $R$). 
We call the subgraph induced by $R\cup F_1\cup F_2$
a {\em trellis\/} (with respect to $R$).
\end{defn} 
One sees that all nodes in $R\cup F_2$ are on the same 2D facet,
whereas some nodes in $F_1$ are also on the facet,
but other nodes in $F_1$ may be out of the facet like a hanging vine. 
Note that removing $R$ from $S$ makes all nodes in the trellis free. 
We define the {\em Trellis-neighborhood of a solution $S$\/}
as the set of all solutions that can be obtained by
removing any $R$ from $S$
and then inserting an independent set among the trellis
into $S\setminus R$. 

\begin{thm}
\label{thm:trellis}
Given a solution $S$,
we can find an improved solution in the Trellis-neighborhood
or conclude that no such solution exists in $O(n^{3.5})$ time. 
\end{thm}

\begin{proof}
We explain how we compute an MIS in the trellis for a given $R$. 
Let us partition $F_1$ into $F_1=F'_1\cup F''_1$
so that $F'_1$ (resp., $F''_1$) is the subset of nodes 
on (resp., out of) the considered 2D facet. 
The node set $F''_1$ induces a subgraph that consists of cliques,
each of which is formed by 1-tight nodes on a grid line perpendicular to the 2D facet. 
One can easily show that, among MISs in the trellis,
there is one such that a solution node is chosen from every clique. 
Intending such an MIS,
we ignore the nodes in $F''_1$
and their solution neighbors; 
let $R''\subseteq R$ be a subset such that $R''=\bigcup_{u\in F''_1}N_S(u)$. 
Now that we have chosen nodes from $F''_1$,
we can no longer choose the nodes in $R''$. 
Let $R'=R\setminus R''$. 
The remaining task is to compute an MIS 
from $R'\cup F'_1\cup F_2$.
All the nodes in $R'\cup F'_1\cup F_2$ are on a 2D facet, 
and we are asked to choose as many nodes as possible 
so that, from each of $2n$ grid lines on the facet, a node is chosen at most once. 
We see that the problem is reduced to computing a maximum matching in a certain bipartite graph;
a grid line on the 2D facet corresponds to a node in the bipartite graph,
and the bipartition of nodes is determined by the directions of grid lines.
A node $v\in R'\cup F'_1\cup F_2$ on the 2D facet corresponds to an edge in the bipartite graph 
such that two nodes are joined 
if $v$ is on the intersecting point of the corresponding two grid lines.  

We have $3n$ 2D facets. 
For each 2D facet, 
it takes $O(n^2)$ time to recognize
the node sets $R$, $F_1$ and $F_2$ and partitions $R=R'\cup R''$ and $F_1=F'_1\cup F''_1$,
and then to construct the bipartite graph  
since the bipartite graph has $2n$ nodes and $O(n^2)$ edges.
We need $O(n^{2.5})$ time to compute a maximum matching~\cite{HK.1973}. 
\QED
\end{proof}

\subsection{Approximation Factors and Computation Time}
\label{sec:ls-time}
In a maximization problem instance,
a {\em $\rho$-approximate solution\/} $(\rho\in[0,1])$
is a solution whose size is at least the factor $\rho$ of the optimal size. 
Hajirasouliha et al.~\cite{HJKS.2007} analyzed
approximation factors of $(p,p+1)$-maximal solutions,
using Hurkens and Schrijver's classical result~\cite{HS.1989}
on the general set packing problem. 
From this and Theorems~\ref{thm:1}, \ref{thm:2}, \ref{thm:3} and \ref{thm:trellis},
and since the solution size is at most $n^2$,
we have the following theorem. 
\begin{thm}
Any 1-maximal, 2-maximal, 3-maximal and Trellis-maximal solutions
are $1/2$-, $5/9$-, $3/5$- and $5/9$-approximate solutions respectively. 
Furthermore, these can be obtained in $O(n^4)$, $O(n^5)$, $O(n^6)$ and $O(n^{5.5})$ time
by extending an arbitrary solution by means of the local search. 
\end{thm}

\section{Computational Study}
\label{sec:exp}

In this section we illustrate how our local search computes
high-quality solutions efficiently 
through a computational study.
We design a prototype iterated local search (ILS) algorithm
and show that it is more likely to deliver better solutions 
than modern optimization solvers.
Furthermore, the best ILS variant is one based on Trellis-swap. 

\subsubsection{ILS Algorithm.}

We describe our ILS algorithm in \algref{ILS},
which is no better than a conventional one. 
Let us describe how to construct the initial solution of the next local search
by ``kicking'' $S^\ast$ (line \ref{line:ILSkick}).
For this, the algorithm copies $S^\ast$ to $S_0$ and
``forces'' to insert $k$ non-solution nodes into $S_0$.
Specifically, it repeats the followings $k$ times;
it chooses a non-solution node $u\notin S_0$,
removes the set $N_{S_0}(u)$ of its solution neighbors from the solution
(i.e., $S_0\gets S_0\setminus N_{S_0}(u)$),
and inserts $u$ into the solution (i.e., $S_0\gets S_0\cup\{u\}$). 
If there appear free nodes, one is chosen at random
and inserted into $S_0$ repeatedly until $S_0$ becomes maximal. 
The number $k$ is set to $k=\kappa$ with probability $1/2^\kappa$. 
We choose the $k$ nodes randomly from all the non-solution nodes,
except the first one. 
We choose the first one with great care so that 
{\bf (i)} trivial cycling is avoided
and {\bf (ii)} the diversity of search
is attained. 
For (i), we restrict the candidates to nodes in $N(S_0')$,
where $S_0'\subseteq S_0$ is a subset of solution nodes such that 
$S_0'=\{x\in S_0\mid \exists d\in[3],\ \mu_d(x)>0\}$. 
In other words, $S_0'$ is a subset of solution nodes
that have a 1-tight neighbor.
Then for (ii), we employ the {\em soft-tabu\/} approach~\cite{ARW.2012};
we choose the non-solution node that has been
outside the solution for the longest time among $N(S_0')$.\footnote{When there is no 1-tight node, $S_0'$ becomes an empty set. In this case, we use $N(S_0)$ instead of $N(S_0')$.}
We omit the details, but the mechanism for (i) dramatically
reduces chances that the next local search returns us 
to the incumbent solution $S^\ast$.  
In fact, it is substantially effective in enhancing the performance of the algorithm.
%

\begin{algorithm}[t!]
\caption{An ILS algorithm for the PLSE problem}
\label{alg:ILS}
\begin{algorithmic}[1]
\State $S_0\gets$an arbitrary solution and $S^\ast\gets S_0$
\Comment{$S^\ast$ is the incumbent solution.}
\While{the computation time does not exceed the given time limit}
\State compute a locally optimal solution $S$ by local search that starts with an initial solution $S_0$
\If{$|S|\ge|S^\ast|$}
\label{line:ILSif}
\State $S^\ast\gets S$
\EndIf
\label{line:ILSendif}
\State compute the initial solution $S_0$ of the next local search by ``kicking'' $S^\ast$
\label{line:ILSkick}
\EndWhile
\State output $S^\ast$
\end{algorithmic}
\end{algorithm}

We consider four variant ILS algorithms
that employ different neighborhoods from each other:
1-ILS, 2-ILS, Tr-ILS and 3-ILS. 
1-ILS is the ILS algorithm that iterates 1-LS 
(i.e., the local search with $(1,n^2)$-neighborhood search algorithm) in the manner of \algref{ILS}.
The terms 2-LS, 2-ILS, Tr-LS, Tr-ILS, 3-LS and 3-ILS are analogous. 

\subsubsection{Experimental Set-up.} 
All the experiments are conducted on a workstation
that carries Intel$^{\textregistered}$ Core$^{\texttrademark}$ i7-4770 Processor (up to 3.90GHz
by means of Turbo Boost Technology)
and 8GB main memory. The installed OS is Ubuntu 14.04.1.

Benchmark instances are random PLSs. 
We generate the instances by utilizing each of the two schemes that are well-known in the literature~\cite{B.2005,GS.2002}:
quasigroup completion (QC) and quasigroup with holes (QWH).  
Note that a PLS is parametrized by the grid length $n$
and the ratio $r\in[0,1]$ of pre-assigned symbols over the $n\times n$ grid. 
Starting with an empty assignment,
QC repeats assigning a symbol to an empty cell randomly so that the resulting assignment is a PLS,
until $\lfloor n^2r\rfloor$ cells are assigned symbols. 
On the other hand, QWH generates a PLS by removing $n^2-\lfloor n^2r\rfloor$ symbols from an arbitrary Latin square
so that $\lfloor n^2r\rfloor$ symbols remain. 
Note that a QC instance does not necessarily
admit a complete Latin square as an optimal solution,
whereas a QWH instance always does. 
%
Here we show experimental results only on QC instances
due to space limitation. 
We note that, however, most of the observed tendencies are quite similar between QC and QWH. 
One can download all the instances
and the solution sizes achieved by the ILS algorithms and by the competitors
from the author's website (\url{http://puzzle.haraguchi-s.otaru-uc.ac.jp/PLSE/}).

Let us mention what kind of instance is ``hard'' in general.
%
Of course an instance becomes harder when $n$ is larger. 
Then we set the grid length $n$ to 40, 50 and 60,
which are relatively large compared with previous studies (e.g.,~\cite{GS.2002}). 
For a fixed $n$,
the problem has easy-hard-easy phase transition. 
Then we regard instances with an intermediate $r$ ``hard''.

For competitors, we employ two exact solvers and one heuristic solver. 
For the former, 
we employ the optimization solver for integer programming model (CPX-IP) 
and the one for constraint optimization model (CPX-CP)
from IBM ILOG CPLEX (ver.~12.6)~\cite{CPLEX}. 
It is easy to formulate the PLSE problem by these models (e.g., see~\cite{GS.2002}).  
For the latter, we employ \LSSOL\ (ver.~4.5)~\cite{LSSOL} (LSSOL), which is
a general heuristic solver based on local search. 
Hopefully our ILS algorithm will outperform LSSOL 
since ours is specialized to the PLSE problem,
whereas LSSOL is developed for general discrete optimization problems. 
All the parameters are set to default values
except that, in CPX-CP,
{\tt DefaultInferenceLevel\/} and {\tt All\allowbreak Diff\allowbreak Inference\allowbreak Level\/} are set to {\tt extended\/}. 
%
We set the time limit of all the solvers 
(including the ILS algorithms) to 30 seconds.

\subsubsection{Results.}

We show how the ILS algorithms and the competitors 
improve the initial solution $S_0$ in \tabref{ils-sol}.\invis{\footnote{The results on Tr-ILS and 3-ILS are different from those in the first submitted version of the paper
  since we removed trivial overheads of their implementation. By the modification, Tr-ILS becomes better and thereby outperforms 2-ILS,
  while the performance of 3-ILS does not change to a large extent.}}
The $S_0$ is generated by a constructive algorithm named G5 in \cite{AKW.2008},
which is a ``look-ahead'' minimum-degree greedy algorithm.
We confirmed in our preliminary experiments
that G5 is the best among several simple constructive algorithms.
For each pair $(n,r)$, 
a number in the 3rd column is the average of $|L|+|S_0|$
(i.e., the given PLS size $|L|=\lfloor n^2r\rfloor$ plus the initial solution size)
and a number in the 4th to 10th columns is the average of
the improved size over 100 instances. 
A bold number (resp., a number with $\ast$)
indicates the 1st largest (resp., the 2nd largest) improvement among all. 
An underlined number indicates that
an optimal solution is found in all the 100 instances.
We can decide a solution $S$ to be an optimal solution
if $L\cup S$ is a complete Latin square or
an exact solver (i.e., CPX-IP or CPX-CP) reports so. 

\begin{table}[t!]
  \centering
  \caption{Improved sizes brought by the ILS algorithms and the competitors}
  \label{tab:ils-sol}
  \begin{tabular}{cc|r|rrrr|rrr}
    \hline
$n$ & $r$ &  $|L|+|S_0|$ & 1-ILS & 2-ILS &  Tr-ILS & 3-ILS &  CPX-IP &  CPX-CP & LSSOL\\
    \hline
40	&	0.3	&	1597.68	&	\underline{\bf2.32}	&	\underline{\bf2.32}	&	\underline{\bf2.32}	&	\underline{\bf2.32}	&	0.00	&	\underline{\bf2.32}	&	$\ast$0.99	\\
	&	0.4	&	1595.28	&	\underline{\bf4.72}	&	\underline{\bf4.72}	&	\underline{\bf4.72}	&	\underline{\bf4.72}	&	0.00	&	\underline{\bf4.72}	&	$\ast$2.72	\\
	&	0.5	&	1591.18	&	\underline{\bf8.82}	&	$\ast$8.80	&	\underline{\bf8.82}	&	$\ast$8.80	&	0.00	&	5.35	&	5.51	\\
	&	0.6	&	1578.89	&	16.64	&	$\ast$17.01	&	\bf17.92	&	16.69	&	5.52	&	5.85	&	11.80	\\
	&	0.7	&	1550.10	&	18.35	&	$\ast$18.76	&	\bf18.78	&	18.33	&	17.37	&	8.68	&	15.01	\\
	&	0.8	&	1508.98	&	$\ast$5.05	&	\underline{\bf5.06}	&	\underline{\bf5.06}	&	5.04	&	\underline{\bf5.06}	&	3.68	&	5.00	\\
\hline																			
50	&	0.3	&	2496.03	&	\underline{\bf3.97}	&	\underline{\bf3.97}	&	\underline{\bf3.97}	&	\underline{\bf3.97}	&	0.00	&	$\ast$3.84	&	0.32	\\
	&	0.4	&	2493.78	&	\underline{\bf6.22}	&	\underline{\bf6.22}	&	\underline{\bf6.22}	&	$\ast$6.20	&	0.00	&	4.24	&	0.87	\\
	&	0.5	&	2488.52	&	11.38	&	$\ast$11.46	&	\underline{\bf11.48}	&	11.34	&	0.00	&	1.40	&	4.44	\\
	&	0.6	&	2476.21	&	20.22	&	$\ast$21.38	&	\bf21.97	&	20.06	&	0.00	&	2.66	&	13.00	\\
	&	0.7	&	2442.21	&	28.03	&	$\ast$28.30	&	\bf28.58	&	27.68	&	4.19	&	8.83	&	21.24	\\
	&	0.8	&	2382.07	&	$\ast$12.18	&	12.14	&	12.16	&	12.08	&	\bf12.51	&	6.03	&	11.60	\\
\hline																			
60	&	0.3	&	3593.07	&	\underline{\bf6.93}	&	\underline{\bf6.93}	&	\underline{\bf6.93}	&	$\ast$6.85	&	0.00	&	5.22	&	0.13	\\
	&	0.4	&	3590.68	&	\underline{\bf9.32}	&	\underline{\bf9.32}	&	\underline{\bf9.32}	&	$\ast$9.20	&	0.00	&	1.87	&	0.49	\\
	&	0.5	&	3585.29	&	$\ast$14.65	&	$\ast$14.65	&	\bf14.67	&	14.06	&	0.00	&	0.54	&	2.21	\\
	&	0.6	&	3572.61	&	24.16	&	\bf25.07	&	$\ast$25.02	&	22.62	&	0.00	&	1.09	&	12.91	\\
	&	0.7	&	3534.62	&	$\ast$37.70	&	37.54	&	\bf39.05	&	35.73	&	0.09	&	5.83	&	26.43	\\
	&	0.8	&	3456.59	&	$\ast$22.15	&	\bf22.16	&	$\ast$22.15	&	21.99	&	21.99	&	7.55	&	19.85	\\
\hline
\end{tabular}
\end{table}

Obviously the ILS algorithms outperform the competitors in many $(n,r)$'-s.
%
We claim that Tr-ILS should be the best  
among the four ILS algorithms. 
Clearly 3-ILS is inferior to others.
The remaining three algorithms seem to be competitive,
but Tr-ILS ranks first or second most frequently.

Concerning the competitors,
CPX-CP performs well for 
under-constrained ``easy'' instances (i.e., $r\le0.4$),
whereas CPX-IP does well for
over-constrained ``easy'' instances (i.e., $r\ge0.7$).
LSSOL is relatively good for all $r$'-s
and outstanding especially for ``hard'' instances with $0.5\le r\le 0.7$. 

To observe the behavior of the ILS algorithms in detail,
we investigate how they improve the solution
in the first LS and in the first 5, 10 and 30 seconds in \tabref{ils-5sec} $(n=60)$.
We also show the averaged computation time of a single run of LS in the rightmost column. 
Most of the improvements over the 30 seconds 
are made in earlier periods. 
It is remarkable that the improvements made by the ILS algorithms
in the first 5 seconds
are larger than those made by the competitors in 30 seconds in all the shown cases (see \tabref{ils-sol}).

\begin{table}[t!]
  \centering
  \caption{Improved sizes in the first LS, 5s, 10s and 30s, and averaged computation times of a single run of LS $(n=60)$;
    when $r$ is smaller (resp., larger),
    the number $|V|$ of nodes in the graph becomes larger (resp., smaller),
    and then LS takes more (resp., less) computation time}
  \label{tab:ils-5sec}
  \begin{tabular}{cc|rlrlrlr|r}
    \hline
$r$ & Algorithm & \multicolumn{7}{c|}{Improved size} & \multicolumn{1}{c}{Computation}\\
& & 1st LS &\ \ \ \ \ 
    &\multicolumn{1}{c}{5s}&\ \ \ \ \ 
    & \multicolumn{1}{c}{10s}&\ \ \ \ \ 
    & \multicolumn{1}{c|}{30s}
    & \multicolumn{1}{c}{time of LS (ms)}\\
    \hline
0.4	&	1-ILS	&	0.29	&&	9.24	&&	9.32	&&	9.32	&	1.44	\\
	&	2-ILS	&	1.03	&&	9.28	&&	9.30	&&	9.32	&	1.84	\\
	&	Tr-ILS	&	1.13	&&	9.28	&&	9.32	&&	9.32	&	2.29	\\
	&	3-ILS	&	2.10	&&	6.73	&&	7.91	&&	9.20	&	42.80	\\
\hline													
0.5	&	1-ILS	&	0.60	&&	14.42	&&	14.58	&&	14.65	&	0.83	\\
	&	2-ILS	&	1.63	&&	14.27	&&	14.47	&&	14.65	&	1.00	\\
	&	Tr-ILS	&	1.65	&&	14.52	&&	14.65	&&	14.67	&	1.04	\\
	&	3-ILS	&	3.08	&&	10.79	&&	12.28	&&	14.06	&	16.82	\\
\hline													
0.6	&	1-ILS	&	0.75	&&	22.58	&&	23.23	&&	24.16	&	0.45	\\
	&	2-ILS	&	2.69	&&	22.78	&&	23.80	&&	25.07	&	0.55	\\
	&	Tr-ILS	&	2.96	&&	23.38	&&	24.06	&&	25.02	&	0.45	\\
	&	3-ILS	&	4.95	&&	17.86	&&	20.18	&&	22.62	&	5.91	\\
\hline													
0.7	&	1-ILS	&	1.57	&&	35.33	&&	36.66	&&	37.70	&	0.24	\\
	&	2-ILS	&	4.33	&&	35.16	&&	36.54	&&	37.54	&	0.28	\\
	&	Tr-ILS	&	5.41	&&	35.90	&&	37.52	&&	39.05	&	0.27	\\
	&	3-ILS	&	7.95	&&	29.08	&&	31.91	&&	35.73	&	2.14	\\
\hline													
\end{tabular}
\end{table}

The reason why Tr-ILS is the best among the ILS variants 
must be its efficiency; 
as can be seen in the rightmost column,
Tr-LS is so fast as 2-LS, or even faster in some $r$'-s,
although Trellis-swap is a generalization of $(2,n^2)$-swap. 
We implemented the Trellis-neighborhood search algorithm 
so that it runs in $O(\alpha n^{2.5})$ time rather than 
in $O(n^{3.5})$ time,
where $\alpha$ denotes the number of 1-tight nodes. 
The implementation is expected to be faster since,
to the extent of our experiment, the number $\alpha$ is much smaller than $n$. 
We will address the detail of this issue in our future papers.

On the other hand, 3-LS is much more time-consuming than the others;
the computation time of 3-LS is about 10 to 40 times
those of the other three LSs.  
3-ILS may not iterate so sufficient a number of 3-LSs
that the diversity of search is not attained to a sufficient level. 
It is true that, in the 1st LS, for all $r$'-s,
3-LS finds the best solution,
followed by Tr-LS, 2-LS and 1-LS; 
a single LS with a larger neighborhood will find a better solution
than one with a smaller neighborhood.  
By the first 5 seconds pass, however,
3-ILS becomes inferior to the other three ILS algorithms. 
The solutions that the three ILS algorithms
find in 10 seconds are better 
than those that 3-ILS outputs after 30 seconds.

Then, to enhance the performance of the ILS algorithm,
the neighborhood size should not be too large. 
It is important to run a fast local search 
with a moderately small neighborhood 
many times from various initial solutions.
We should develop a better mechanism to
generate a good initial solution of the local search
rather than to investigate a larger neighborhood. 
This is left for future work. 

\section{Concluding Remarks}
\label{sec:conc}

We have designed efficient local search algorithms for the PLSE problem
such that the neighborhood is defined by swap operation. 
The proposed $(p,n^2)$-neighborhood search algorithm $(p\in\{1,2,3\})$
finds an improved solution in the neighborhood
or concludes that no such solution exists in $O(n^{p+1})$ time.
We also proposed a novel swap operation, Trellis-swap,
a generalization of $(1,n^2)$-swap and $(2,n^2)$-swap,
whose neighborhood search algorithm takes $O(n^{3.5})$ time. 

Our achievement is attributed to observation on the graph structure such that
each node is regarded as a 3D integral point, 
its neighbors are partitioned into $O(1)$ cliques
and no two neighbors in different cliques are adjacent. 
Our idea is never limited to the PLSE problem
but can be extended to MIS problems on graphs having the same structure,
including some instances of the {\em maximum strong independent set\/} problem
on hypergraphs~\cite{HL.2009}. 

Our ILS algorithm is no better than a prototype
and has much room for improvement. 
Nevertheless it outperforms IBM ILOG CPLEX and \LSSOL\ in most of the tested instances. 
Of course there are various solvers available,
and comparison with them is left for future work. 
Among these, we consider that SAT based solvers may not be so effective
due to our preliminary experiments as follows;
we tried to solve the satisfiability problem on QWH instances 
by \SUGAR\ (ver.~2.2.1)~\cite{SUGAR},
where we used \MINISAT\ (ver.~2.2.1)~\cite{MINISAT} as the core SAT solver. 
Note that any QWH instance is satisfiable.
\SUGAR\ decides the satisfiability (and thus finds an optimal solution)
for about $50\%$ of the instances
($n=40$ and $r\in\{0.3,0.4,...,0.8\}$) within 30 seconds,
while Tr-ILS finds an optimal solution for $78\%$ 
of the instances within the same time limit.  

Alternatively one can conceive another metaheuristic algorithm,
utilizing methodologies developed so far~\cite{GK.2003,G.2007}. 
Our local search can be a useful tool for this. 
For example, in GA,
one can enhance the quality of a population
by performing our efficient local search on each solution. 

Although the local search achieves
the best approximation factor for the PLSE problem currently,
no one has explored its efficient implementation in the literature.
This work resolves this issue to some degree.  

\newpage
\bibliographystyle{splncs03}
\bibliography{mybib}

\end{document}